\documentclass[letterpaper, 10 pt, conference]{ieeeconf}
\IEEEoverridecommandlockouts
\usepackage{cite}
\usepackage{algorithmic}
\usepackage{graphicx}
\usepackage{textcomp}
\usepackage{xcolor}
\usepackage{tikz}
\usetikzlibrary{arrows.meta,positioning,fit,calc,shapes.misc}
\usetikzlibrary{positioning,calc}
\usepackage[caption=false,font=footnotesize]{subfig}
\usepackage{mathtools}
\newtheorem{theorem}{Theorem}

\usepackage{newtxtext}
\usepackage{newtxmath}
\usepackage{microtype}

\usepackage{hyperref}

\hypersetup{ 
pdftitle={Control-Theoretic View of Neural ODEs: Empirical Controllability and Observability},   
 pdfsubject={Neural Ordinary Differential Equations}, 
 pdfkeywords={Neural Ordinary Differential Equations, controllability, observability},
 colorlinks=true,
 filecolor=green, 
 citecolor=red,       
 urlcolor=blue}
\usepackage{eso-pic}
 \def\BibTeX{{\rm B\kern-.05em{\sc i\kern-.025em b}\kern-.08em
    T\kern-.1667em\lower.7ex\hbox{E}\kern-.125emX}}
    
\title{\LARGE \bf
Control-Theoretic View of Neural ODEs: Empirical Controllability and Observability
}

\author{Md Saiful Islam and Rahul Bhadani  % <-this % stops a space
\thanks{\scriptsize The authors are with AI, Autonomy, Resilience, Control (AARC) Lab at the Department of Electrical and Computer Engineering, The University of Alabama in Huntsville, Alabama, USA. Email: {\textit{mi1499@uah.edu}, \textit{rahul.bhadani@uah.edu}}}%
}

\begin{document}
\AddToShipoutPictureFG*{%
  \AtPageLowerLeft{%
    \raisebox{0.5in}{%
      \hspace*{2.5in}%
      {\footnotesize This manuscript has been submitted to IEEE Conference on Decision and Control 2026}%
    }%
  }%
}
\maketitle
\thispagestyle{empty}
\pagestyle{empty}

%\begin{abstract}
%This paper studies neural ordinary differential equations (neural ODEs) from a control-theoretic perspective using controllability and observability concepts. The neural ODE is modeled in a control-affine form to enable analysis using tools from nonlinear and linear time-varying (LTV) systems. Controllability is analyzed through trajectory linearization, where the LTV controllability Gramian provides a first-order, local measure of state influence along a nominal trajectory. Observability is studied using output linearization, where the LTV observability Gramian characterizes local state reconstructability from output measurements along a nominal trajectory. Koopman-based lifting is discussed as a framework to extend the analysis to a higher-dimensional linear representation, with limitations identified under multiple equilibria and basin-dependent behavior. The proposed framework is demonstrated on a series RLC circuit, where the learned neural ODE maintains nearly accurate trajectory reconstruction under Gaussian measurement noise and generalizes to unseen initial conditions. Empirical Gramian analysis confirms positive eigenvalues consistent with local controllability and observability of the learned dynamics. Distinct initial conditions produce distinguishable output trajectories, supporting the Gramian-based observability interpretation.
%\end{abstract}

\begin{abstract}
This paper studies neural ordinary differential equations (neural ODEs) from a control-theoretic perspective using controllability and observability concepts. The neural ODE is represented in a control-affine form to facilitate analysis using tools from nonlinear and linear time-varying (LTV) systems. Controllability is examined through trajectory linearization, where the LTV controllability Gramian provides a local, first-order measure of input influence along a nominal trajectory. Observability is analyzed through output linearization, where the LTV observability Gramian characterizes the local ability to reconstruct system states from output measurements. Koopman-based lifting is considered to extend the analysis to a higher-dimensional representation, and its limitations under multiple equilibria and basin-dependent behavior are discussed. The proposed framework is illustrated on a series RLC circuit. The learned neural ODE reproduces system trajectories and generalizes to unseen initial conditions. The computed Gramians are numerically full rank along the tested trajectories, indicating local controllability and observability of the linearized dynamics. 
\end{abstract}

%\begin{IEEEkeywords}
 %   Neural ordinary differential equations, controllability, observability, Koopman lifting.
%\end{IEEEkeywords}

\section{Introduction}

Neural ordinary differential equations (neural ODEs) provide a framework for learning continuous-time system dynamics from data. When neural networks are modeled as differential equations, they represent network depth as continuous time, enabling flexible architectures, reduced memory usage, and effective modeling of nonlinear dynamical systems \cite{chen2018neural, yang2025augmented}. Despite these advantages, neural ODEs are primarily studied from a learning perspective and are often treated as black-box models in control and estimation. As a result, basic system-theoretic properties such as controllability and observability are not explicitly analyzed. This lack of analysis limits their interpretability and reliability, particularly in safety-critical applications where guarantees on state reachability and reconstruction are important \cite{niu2024applications, rahman2022neural, martinelli2023unconstrained}.

Existing work on neural ODEs focuses on model expressiveness, numerical stability, and efficient training methods. This includes techniques like adjoint-based gradient computation and solver-aware training approaches \cite{chen2018neural, gholami2019anode}. At the intersection of learning and control, prior work has explored neural state-space models and data-driven system identification \cite{zhang2026neural}. However, those approaches consider controllability and observability implicitly rather than analyzing these properties directly using system-theoretic tools. Whereas classical control theory provides well-established criteria for controllability and observability based on linearization and Gramian analysis for nonlinear and linear time-varying systems \cite{verriest2025comparison, cheng2025control}. However, these tools are not systematically applied to neural ODE models.

Koopman operator theory offers a linear representation of nonlinear dynamical systems by lifting the state into a higher-dimensional space of observables, where the dynamics evolve linearly. Koopman theory provides a powerful framework for system analysis, model reduction, and control of nonlinear systems \cite{shi2026koopman}. While Koopman theory has been used to study controllability and observability in general nonlinear systems~\cite{liu2025identification}, its application to neural ODE architectures, especially for analyzing classical system-theoretic properties, remains less explored.

Building on these observations, we adopt a control-theoretic perspective on neural ODEs. We model neural ODEs as nonlinear control systems with explicit inputs and analyze their dynamical properties using trajectory linearization and system-theoretic tools as illustrated in Figure~\ref{fig:1}. In particular, we analyze local controllability and observability using linear time-varying Gramians, and extend the analysis via Koopman lifting under suitable assumptions.

%%%%%%%%%%%%%%%%%%%%%%%%%%%%%%%%%%%%%%%%%%%%%%%%%%%%%%%%%%%%%%
%---------------- Figure ----------------%
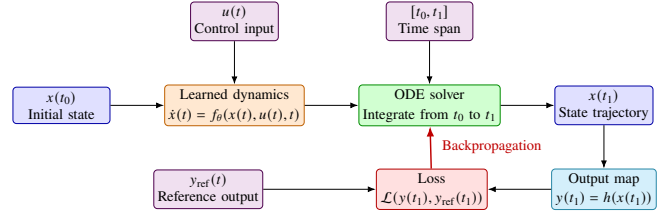
\begin{figure}[t!]
\centering
\resizebox{1.0\linewidth}{!}{
\begin{tikzpicture}[
    node distance=1.35cm and 1.55cm,
    >=Latex,
    every node/.style={font=\large},
    block/.style={
        draw,
        rounded corners=3pt,
        thick,
        align=center,
        minimum height=1.05cm,
        minimum width=2.8cm
    },
    io/.style={
        block,
        fill=blue!12,
        draw=blue!65!black
    },
    dyn/.style={
        block,
        fill=orange!18,
        draw=orange!75!black,
        minimum width=3.4cm
    },
    solver/.style={
        block,
        fill=green!18,
        draw=green!55!black,
        minimum width=3.0cm
    },
    output/.style={
        block,
        fill=cyan!15,
        draw=cyan!60!black,
        minimum width=2.6cm
    },
    loss/.style={
        block,
        fill=red!12,
        draw=red!70!black,
        minimum width=3.0cm
    },
    aux/.style={
        block,
        fill=violet!12,
        draw=violet!60!black,
        minimum width=2.5cm
    },
    arrow/.style={->, thick},
    dashedarrow/.style={->, thick, dashed, gray!80},
    backarrow/.style={->, very thick, red!75!black}
]

% Main forward path
\node[io] (x0) {$x(t_0)$\\Initial state};

\node[dyn, right=of x0] (dynamics)
{Learned dynamics\\[2pt]
$\dot{x}(t)=f_\theta(x(t),u(t),t)$};

\node[solver, right=of dynamics] (odesolver)
{ODE solver\\[2pt]
Integrate from $t_0$ to $t_1$};

\node[io, right=of odesolver] (xt)
{$x(t_1)$\\State trajectory};

\node[output, below=1.25cm of xt] (yblock)
{Output map\\[2pt]
$y(t_1)=h(x(t_1))$};

\node[loss, left=1.83cm of yblock] (lossblock)
{Loss\\[2pt]
$\mathcal{L}\!\left(y(t_1),y_{\mathrm{ref}}(t_1)\right)$};

\node[aux, left=3.25cm of lossblock, xshift=0cm] (yref)
{$y_{\mathrm{ref}}(t)$\\Reference output};

% Inputs
\node[aux, above=1.25cm of dynamics] (u)
{$u(t)$\\Control input};

\node[aux, above=1.25cm of odesolver] (time)
{$[t_0,t_1]$\\Time span};

% Optional measured/output trajectory
%\node[aux, below=2.5cm of odesolver, xshift=1.7cm] (traj)
%{$x(t),\,y(t)$\\Generated trajectory};

% Forward arrows
\draw[arrow] (x0) -- (dynamics);
\draw[arrow] (dynamics) -- node[above] {} (odesolver);
\draw[arrow] (odesolver) -- (xt);
\draw[arrow] (xt) -- (yblock);
\draw[arrow] (yblock) -- (lossblock);
\draw[arrow] (yref) -- (lossblock);
\draw[arrow] (u) -- (dynamics);
\draw[arrow] (time) -- (odesolver);

% Trajectory annotation
%\draw[dashedarrow] (odesolver) -- (traj);
%\draw[dashedarrow] (traj.west) -| ([yshift=-5pt]dynamics.south);

% Backward adjoint flow
\draw[backarrow] (lossblock) -- 
node[pos=0.25, above, xshift=50pt, text=red!75!black] {Backpropagation} 
(odesolver);

%\draw[backarrow] ([yshift=7pt]odesolver.south) -- ([yshift=7pt]dynamics.north);

\end{tikzpicture}
}
\caption{
Control-oriented neural ODE training architecture. The learned dynamics are modeled as
$\dot{x}(t)=f_\theta(x(t),u(t),t)$, where the neural network parameterizes the state equation and the control input $u(t)$ acts explicitly on the dynamics. An ODE solver integrates the system from the initial condition $x(t_0)$ over the interval $[t_0,t_1]$ to generate the state trajectory $x(t_1)$, which is mapped to the output $y(t_1)=h(x(t_1))$. Training is performed by minimizing a loss $\mathcal{L}(y(t_1),y_{\mathrm{ref}}(t_1))$, with gradients computed via automatic differentiation through the ODE solver.
%\rkb{ODE Solver is a part of Neural ODE, so an accurate block diagram should show a Neural Network block, and an initial condition as input to the solver. Essentially, ODESolver is a part of the Neural Network's forward function. This needs to be redone.}
}
\label{fig:1}
\end{figure}
%%%%%%%%%%%%%%%%%%%%%%%%%%%%%%%%%%%%%%%%%%%%%%%%%%%%%%%%%%%%%%

The results in this paper are inherently local and conditional. Gramian-based rank conditions are exact for the linearized dynamics but provide only local indicators of accessibility and state reconstruction for the underlying nonlinear system. Similarly, Koopman-based analysis depends on assumptions such as injective lifting and finite-dimensional closure, and can not hold in the presence of multiple equilibria or basin-dependent behavior.
This work does not propose a new learning algorithm or controller. Instead, it provides a theoretical framework for interpreting learned neural ODE models through the lens of nonlinear systems theory, clarifying when such models admit meaningful control and estimation interpretations.

\subsection{Contributions}

This paper provides a control-theoretic interpretation of neural ODEs. The main contributions are:\\
\textbf{(1) Controlled neural ODEs:} Neural ODEs are modeled as nonlinear control systems with explicit inputs in a control-affine form, allowing the use of nonlinear and linear time-varying analysis tools.\\
\textbf{(2) Local controllability and observability:} Trajectory linearization is used to construct linear time-varying controllability and observability Gramians, which provide local measures of input influence and state reconstruction.\\
\textbf{(3) Koopman-based analysis:} Koopman lifting is used to extend the analysis beyond local linearization, while noting limitations such as multi-equilibria and basin-dependent behavior.\\
\textbf{(4) Empirical validation:} The framework is evaluated on a series RLC circuit. The results are consistent with local controllability and observability, and the model shows consistent behavior across tested initial conditions.

%k-------------------------------------------------------------
\section{Neural Ordinary Differential Equations}
%k--------------------------------------------------------------
%\subsection{From Residual Networks to Continuous Depth}

Neural ODEs generalize residual networks (ResNets) by treating network depth as continuous time, making the model behave like a dynamical system, presented in Figure \ref{fig:1}. ResNets evolve the state in discrete steps (layer-by-layer), whereas neural ODEs evolve the state continuously over time ~\cite{chen2018neural}. ResNets update each layer by adding a small Forward-Euler step as:
\begin{equation}
x_{k+1} = x_k + h f(x_k,\theta_k)
\end{equation}
As the step size $h \to 0$, the discrete-depth network approaches a continuous-depth model, where the hidden state evolves according to a neural ODE \cite{chen2018neural}:
\begin{equation}\label{eq:node_base}
\dot{x}(t) = f_\theta\!\big(x(t),u(t),t\big), \qquad x(t_0)=x_0.
\end{equation}
For control-theoretic analysis, we parameterize the controlled neural ODE in a control-affine form:
\begin{equation}\label{eq:control_affine}
\dot{x}(t) =
f_\theta\!\big(x(t),t\big) +
G_\theta\!\big(x(t),t\big)\, u(t)
\end{equation}
where $x(t)\in\mathbb{R}^n$ is the state and $u(t)\in\mathbb{R}^m$ is the control input. $f_\theta:\mathbb{R}^n\times\mathbb{R}\to\mathbb{R}^n$ and
$G_\theta:\mathbb{R}^n\times\mathbb{R}\to\mathbb{R}^{n\times m}$ are represented by neural networks.
Equation~\eqref{eq:control_affine} reflects a modeling assumption rather than a general property of arbitrary $f_\theta(x,u,t)$.
Equivalently, $G_\theta=[g_1,\dots,g_m]$ with $g_i\in\mathbb{R}^n$ provides:
\begin{equation}\label{eq:control_affine_columns}
\dot{x}(t)=f_\theta\!\big(x(t),t\big)+\sum_{i=1}^{m} g_i\!\big(x(t),t,\theta\big)\,u_i(t)
\end{equation}
The network output is obtained by integrating the continuous dynamics from the initial time $t_0$ to the terminal time $t_1$ as:
\begin{equation}
x(t_1) = \mathrm{ODESolve}\!\left(x(t_0), f_\theta, t_0, t_1 \right)
\end{equation}

\section{Differential Controllability for Neural ODEs}
%%%%%%%%%%%%%%%%%%%%%%%%%%%%%%%%%%%%%%%%%%%%%%%%%%%%%%%%%%%%%%%%%%%%%%%%%%%%

To study controllability, neural ODEs are augmented with explicit control inputs $u_i$, enabling the use of nonlinear and linear control tools.

\subsection{Differential Controllability via Trajectory Linearization}

%\rkb{nominal trajectory = ideal trajectory?}

Consider a nominal trajectory $(x^*(t),u^*(t))$ indicating a particular solution of the nonlinear system under a chosen input, which serves as the operating point about which dynamics are locally linearized. Linearizing Equation \eqref{eq:node_base} about the nominal trajectory yields the linear time-varying (LTV) system as:

\begin{equation}\label{eq:lti_lin_ctrl}
\delta \dot{x}(t) = A(t)\,\delta x(t) + B(t)\,\delta u(t)
\end{equation}
where
$A(t) = \left.\frac{\partial f}{\partial x}\right|_{x^*,u^*,t}
\quad \text{and, } \quad
B(t) = \left.\frac{\partial f}{\partial u}\right|_{x^*,u^*,t}$.
The derivative of the state-transition matrix $\Phi(t,\tau)$ is:
\begin{equation}
\frac{d}{dt}\Phi(t,\tau) = A(t)\Phi(t,\tau), \qquad \Phi(\tau,\tau)=I
\end{equation}
Using the controllability Gramian, we evaluate whether the system can be driven to a desired state by the input~\cite{kalman1963mathematical}. The controllability Gramian over $[t_0,t]$ is:
\begin{equation}\label{eq:Wc}
W_c(t_0,t) =
\int_{t_0}^{t}
\Phi(t,\tau)\, B(\tau)B(\tau)^\top\, \Phi(t,\tau)^\top \, d\tau 
\end{equation}
\textbf{Assumption 1:}
$f(x,u,t,\theta)$ is continuously differentiable in $(x,u)$ $\forall\, t$, and $A(t),B(t)$ are bounded and piecewise continuous on $[t_0,t]$.

In addition, it is assumed that $A(\cdot), B(\cdot)$ are continuous or piecewise continuous so that $\Phi$ exists and the standard LTV Gramian rank test applies.

\begin{theorem}[LTV Controllability Gramian Rank Test]\label{thm:lvt_ctrl}
Under Assumption~1, the LTV system in Equation~\eqref{eq:lti_lin_ctrl}
is controllable on $[t_0,t_1]$ if and only if the controllability Gramian in Equation~\eqref{eq:Wc}
satisfies $\mathrm{rank}\big(W_c(t_0,t_1)\big)=n$.
\end{theorem}
\begin{proof}
Let $L^2([t_0,t_1];\mathbb{R}^m)$ be the space of square-integrable input
perturbations. The reachability operator
$\mathcal{R}:L^2([t_0,t_1];\mathbb{R}^m)\rightarrow\mathbb{R}^n$ can be defined as:
\begin{equation}
\mathcal{R}(\delta u)=\int_{t_0}^{t_1}\Phi(t_1,\tau)B(\tau)\delta u(\tau)\,d\tau 
\end{equation}
The reachable set of the system at time $t_1$ equals $\mathrm{Range}(\mathcal{R})$.
Hence, the system is controllable if and only if
$\mathrm{Range}(\mathcal{R})=\mathbb{R}^n$.
At this stage, we compute the adjoint operator $\mathcal{R}^\ast$.
For any $v\in\mathbb{R}^n$ and $\delta u(\cdot)\in L^2$:
\begin{align}
\langle \mathcal{R}(\delta u),v\rangle_{\mathbb{R}^n}
&=v^\top \int_{t_0}^{t_1}\Phi(t_1,\tau)B(\tau)\delta u(\tau)\,d\tau \nonumber \\
&=\int_{t_0}^{t_1}\delta u(\tau)^\top B(\tau)^\top
\Phi(t_1,\tau)^\top v\,d\tau 
\end{align}
Therefore,
$(\mathcal{R}^\ast v)(\tau)=B(\tau)^\top\Phi(t_1,\tau)^\top v,
\qquad \tau\in[t_0,t_1]$.
Applying $\mathcal{R}$ to $\mathcal{R}^\ast v$ provides:
\begin{align}
(\mathcal{R}\mathcal{R}^\ast v)
&=\int_{t_0}^{t_1}\Phi(t_1,\tau)B(\tau)(\mathcal{R}^\ast v)(\tau)\,d\tau \nonumber \\
&=\int_{t_0}^{t_1}\Phi(t_1,\tau)B(\tau)B(\tau)^\top
\Phi(t_1,\tau)^\top v\,d\tau  \nonumber \\
&=W_c(t_0,t_1)\,v
\end{align}
Hence,
$\mathcal{R}\mathcal{R}^\ast=W_c(t_0,t_1)$.
The operator $\mathcal{R}\mathcal{R}^\ast$ is positive semidefinite, and the reachable set satisfies $\mathrm{Range}(\mathcal{R})=\mathbb{R}^n$ if and only if $\mathcal{R}\mathcal{R}^\ast$ is nonsingular.
Since $\mathcal{R}\mathcal{R}^\ast=W_c(t_0,t_1)$, the system is controllable
if and only if $W_c(t_0,t_1)$ is full rank, i.e.,
\[
\mathrm{rank}\big(W_c(t_0,t_1)\big)=n.
\]
\end{proof}

\textbf{Remark 1:}
Theorem~\ref{thm:lvt_ctrl} is exact for the linearized LTV system. For the original nonlinear neural ODE, the condition $\mathrm{rank}(W_c(t_0,t))=n$ should be interpreted as a first-order local accessibility indicator along the nominal trajectory $(x^*(t),u^*(t))$, rather than a global controllability guarantee.

%For a nonlinear control-affine system $\dot x=f_0(x)+\sum_{i=1}^m g_i(x)u_i$, the full-rank linearized Gramian is an infinitesimal property along a trajectory. Stronger notions such as small-time local controllability (STLC) generally require additional nonlinear conditions \rkb{In my proposal submitted last year, I  mentioned this to explore, you can put this as the next TODO} (e.g., Lie-bracket/Lie-algebra rank conditions and their refinements). In this paper, we therefore use $W_c$ primarily as a trajectory-dependent, first-order system-theoretic metric of input influence.

%%%%%%%%%%%%%%%%%%%%%%%%%%%%%%%%%%%%%%%%%%%%%%%%%%%%%%%%%%%%%%%%%%%%%%%%%%%%
\subsection{Koopman-Based Lifting for Controlled Neural ODEs}
%%%%%%%%%%%%%%%%%%%%%%%%%%%%%%%%%%%%%%%%%%%%%%%%%%%%%%%%%%%%%%%%%%%%%%%%%%%%

The previous analysis applies only near one operating point. Koopman theory lifts the nonlinear system into a feature space where it behaves approximately linear, allowing the use of methods from linear-system analysis \cite{hirano2026conformal}.

\noindent \textbf{Assumption 2:}
A lifting map $z=\psi(x)\in\mathbb{R}^N$ transforms the state into a higher-dimensional space where the dynamics are approximately bilinear as:
$\dot z(t)=A_K z(t)+\sum_{i=1}^m B_{K,i}z(t)u_i(t)+r(t)$, where $A_K,B_{K,i}\in\mathbb{R}^{N\times N}$ are constant matrices, and $r(t)$ is a bounded residual term satisfying $\|r(t)\|\le \bar r$.

\noindent\textbf{Remark 2:}
If $f$ depends explicitly on time, Koopman representations will require time-varying matrices $A_K(t), B_{K,i}(t)$ or an augmented state $(x,t)$ with $\dot t=1$ and lifting $z=\psi(x,t)$.

\noindent \textbf{Assumption 3:}
The nonlinear controlled system with disturbance is:
$\dot{x}(t)=f(x(t),u(t),t,\theta)+\sigma(t)$,
$ x\in\mathbb{R}^n$,
in which the lumped disturbance is bounded as $\|\sigma(t)\|\le \zeta$ for some
$\zeta>0$.

According to \textit{Assumption 2}, $z=\psi(x)\in\mathbb{R}^N$ is a continuously differentiable lifting map including Jacobian $J_\psi(x)=\frac{\partial \psi}{\partial x}$. Then, the lifted dynamics can be written as:
\begin{equation}\label{eq:koop_disturbed}
\dot{z}(t)=A_K z(t)+\sum_{i=1}^m B_{K,i}\,z(t)\,u_i(t)+r(t)+\sigma_K(t)
\end{equation}
where the lifted disturbance is
$\sigma_K(t)=J_\psi(x(t))\,\sigma(t)$.
On the region of interest $\mathcal{X}$, the Jacobian is uniformly bounded, i.e., $\sup_{x\in\mathcal{X}}\|J_\psi(x)\|\le \bar{J}<\infty$, such that $\|\sigma_K(t)\|\le \bar{J}\,\zeta $.

\subsection{Controllability in Koopman Space}

By locally linearizing the bilinear Koopman model in Equation \eqref{eq:koop_disturbed} in the lifted coordinates, we can apply linear controllability tests. Around the nominal trajectory$(z^*(t),u^*(t))$ we obtain:
\begin{equation}\label{eq:koop_ltv_lin}
\delta \dot z(t)=A_K^{\mathrm{lin}}(t)\delta z(t)+B_K^{\mathrm{lin}}(t)\delta u(t)
\end{equation}
where $A_K^{\mathrm{lin}}(t)=A_K+\sum_{i=1}^m B_{K,i}u_i^*(t)$ and
$B_K^{\mathrm{lin}}(t)=[B_{K,1}z^*(t)\;\cdots\;B_{K,m}z^*(t)]$.

Now, if $\Phi_K(t,\tau)$ is the state-transition matrix of $A_K^{\mathrm{lin}}(t)$, then the Koopman controllability Gramian is as:
\begin{equation}\label{eq:WcK}
W_c^{(K)}(t_0,t)=\int_{t_0}^t
\Phi_K(t,\tau)\, B_K^{\mathrm{lin}}(\tau)\big(B_K^{\mathrm{lin}}(\tau)\big)^\top
\Phi_K(t,\tau)^\top\, d\tau 
\end{equation}

\textbf{Remark 3:}
If $\mathrm{rank}(W_c^{(K)}(t_0,t))=N$, then the linearized lifted system in Equation \eqref{eq:koop_ltv_lin} is controllable on $[t_0,t]$ in $\mathbb{R}^N$. Under a sufficiently expressive and locally injective lifting map $\psi(\cdot)$, this indicates local, trajectory-dependent approximate controllability of the original nonlinear state within the region $\mathcal{X}$.
%%%%%%%%%%%%%%%%%%%%%%%%%%%%%%%%%%%%%%%%%%%%%%%%%%%%%%%%%%%%%%%%%%%%%%%%%%%%

%%%%%%%%%%%%%%%%%%%%%%%%%%%%%%%%%%%%%%%%%%%%%%%%%%%%%%%%%%%%%%%%%%%%%%%%%%%%
\subsection{Effect of Multiple Equilibria and Basin-Dependent Properties}
%%%%%%%%%%%%%%%%%%%%%%%%%%%%%%%%%%%%%%%%%%%%%%%%%%%%%%%%%%%%%%%%%%%%%%%%%%%%

Neural ODEs are nonlinear dynamical systems whose vector field is represented by a neural network. 
For a constant input $u_e$, equilibrium points $x_i^*$ satisfying $f(x_i^*,u_e,\theta)=0$. Thus, when the input remains constant at $u_e$, the state remains unchanged if initialized at $x_i^*$. Because $f(\cdot)$ is nonlinear, multiple such equilibrium points may exist.

\textbf{Controllability:}
The Gramian analysis in the previous section is based on linearization along a nominal trajectory $(x^*(t),u^*(t))$. 
The condition $\mathrm{rank}(W_c(t_0,t))=n$ guarantees controllability of the linearized system near that trajectory. 
However, in the nonlinear neural ODE, the reachable set depends on the basin containing the initial condition. 
Moving between basins can require passing through unstable regions. Therefore, local controllability does not imply global controllability.

\textbf{Observability:}
For the measurement model: $y(t)=h(x(t),t)$, distinct states from different basins can produce identical outputs. 
Consequently, though the observability Gramian is full rank, the state can be reconstructed only locally near the nominal trajectory. 
Thus, observability is also trajectory dependent.

\textbf{Koopman lifting:}
Koopman analysis relies on a lifting $z=\psi(x)$. 
If two different states satisfy $\psi(x_a)=\psi(x_b)$, the lifted coordinates cannot uniquely represent the original state. 
Therefore, lifted observability provides state distinguishability only under an injective lifting
$\psi:\mathcal{X}\to\mathbb{R}^N$.

\textbf{Disturbance effects:}
Under disturbances $\dot{x}(t)=f(x(t),u(t),t,\theta)+\sigma(t)$, the lifted disturbance is $\sigma_K(t)=J_\psi(x(t))\sigma(t)$. 
If $\|\sigma(t)\|\le \zeta$ and the lifting Jacobian is bounded as $\|J_\psi(x)\|\le \bar J$ on the region of interest, then
$\|\sigma_K(t)\|\le \bar J\,\zeta$.
Thus, although the disturbance is bounded in the original state, its magnitude in the lifted space is state-dependent and can be amplified near unstable regions.

\noindent\textbf{Remark 4:}
With multiple equilibria, the state space splits into basins of attraction. Therefore, the controllability and observability Gramians describe behavior only near the nominal trajectory. A full-rank Gramian indicates local accessibility and local state reconstruction, not global controllability or observability.

\noindent\textbf{Remark 5:}
Koopman lifting $z=\psi(x)$ must be injective on the region of interest. Otherwise, different states can share the same lifted coordinate, and lifted observability does not provide true state observability. 
In addition, the lifted disturbance $\sigma_K(t)=J_\psi(x(t))\sigma(t)$ depends on the state, so disturbances can be amplified near unstable regions.

%Overall, Neural ODEs with multiple equilibria behave as multi-operating-point nonlinear systems. Accordingly, controllability and observability results in this paper should be interpreted as local properties along trajectories rather than global guarantees.

%%%%%%%%%%%%%%%%%%%%%%%%%%%%%%%%%%%%%%%%%%%%%%%%%%%%%%%%%%%%%%%%%%%%%%%%%%%%
\section{Observability Analysis of Neural ODEs}
%%%%%%%%%%%%%%%%%%%%%%%%%%%%%%%%%%%%%%%%%%%%%%%%%%%%%%%%%%%%%%%%%%%%%%%%%%%%

Observability determines whether the internal state of a neural ODE can be reconstructed from its outputs.

\subsection{Local Observability via Linearization}

Assume the measurement equation is:
\begin{equation}\label{eq:meas}
y(t) = h(x(t),t), \qquad y \in \mathbb{R}^p
\end{equation}
The output function can itself be a neural network, making observability dependent on learned representations.
Linearizing Equation \eqref{eq:meas} about a nominal trajectory $x^*(t)$ provides:
\begin{equation}\label{eq:lin_output}
\delta y(t) = C(t)\,\delta x(t),
\qquad
C(t) = \left.\frac{\partial h}{\partial x}\right|_{x^*,t}
\end{equation}
The observability Gramian is the ability to determine the state of the linear system via its output in finite time~\cite{kalman1963mathematical}. The observability Gramian over $[t_0,t]$ is
\begin{equation}\label{eq:Wo}
W_o(t_0,t) =
\int_{t_0}^{t}
\Phi(\tau,t_0)^\top C(\tau)^\top C(\tau)\,
\Phi(\tau,t_0)\, d\tau 
\end{equation}

\textbf{Assumption 4:} $h(x,t)$ is continuously differentiable in $x$ for all $t\in[t_0,t]$, and $C(t)$ is bounded and piecewise continuous on $[t_0,t]$.

\begin{theorem}[LTV Gramian Rank Test for Observability]\label{thm:lvt_obs}
Under Assumptions~1 and~4, the linearized system $\delta\dot x=A(t)\delta x$, $\delta y=C(t)\delta x$
is observable on $[t_0,t]$ if and only if $\mathrm{rank}(W_o(t_0,t))=n$.
\end{theorem}\vspace{-10pt}
\begin{proof}
Observability of $(A(t),C(t))$ is equivalent to controllability of the dual system $\dot{\eta}(t)= -A(t)^\top \eta(t)+C(t)^\top v(t)$. Applying Theorem~\ref{thm:lvt_ctrl} to the dual system gives the Gramian rank condition.
\end{proof}

\noindent\textbf{Remark 6:}
Theorem~\ref{thm:lvt_obs} is exact for the linearized LTV model and provides a first-order local observability test for the original nonlinear neural ODE.

%%%%%%%%%%%%%%%%%%%%%%%%%%%%%%%%%%%%%%%%%%%%%%%%%%%%%%%%%%%%%%%%%%%%%%%%%%%%

\subsection{Koopman-Based Observability Analysis}
%%%%%%%%%%%%%%%%%%%%%%%%%%%%%%%%%%%%%%%%%%%%%%%%%%%%%%%%%%%%%%%%%%%%%%%%%%%%

Koopman lifting enables observability analysis in an enlarged coordinate space.

\textbf{Assumption 5:}
The state is lifted using a function $z=\psi(x)\in\mathbb{R}^N$. The output in the lifted space is
$y(t)=C_K z(t)$
for some $C_K\in\mathbb{R}^{p\times N}$ on the region of interest $\mathcal{X}$.

\textbf{Assumption 6:}
$\psi:\mathcal{X}\to\mathbb{R}^N$ is injective, so $x$ is uniquely determined by $z=\psi(x)$ on $\mathcal{X}$.

%\subsubsection{Koopman Observability Gramian}

For the autonomous lifted linear system $\dot z = A_K z$, $y=C_K z$, the Koopman observability Gramian over $[t_0,t]$ is
\begin{equation}\label{eq:WoK}
W_o^{(K)}(t_0,t) =
\int_{t_0}^{t}
e^{A_K^\top(\tau-t_0)}\, C_K^\top C_K \,
e^{A_K(\tau-t_0)} \, d\tau 
\end{equation}

\begin{theorem}[Koopman Observability Implies State Observability on $\mathcal{X}$]
\label{thm:koop_obs}
Suppose Assumptions~5 and 6 hold on $\mathcal{X}$, and consider the autonomous lifted output model $\dot z=A_K z$, $y=C_K z$ (or the time-augmented construction in Remark~2). If $(A_K,C_K)$ is observable on $[t_0,t]$ (equivalently,
$\mathrm{rank}(W_o^{(K)}(t_0,t))=N$), then the original state $x(t_0)\in\mathcal{X}$ is uniquely determined from the output trajectory $y(\cdot)$ on $[t_0,t]$.
\end{theorem}\vspace{-10pt}
\begin{proof}
Observability of $(A_K,C_K)$ implies that $z(t_0)$ is uniquely determined by $y(\cdot)$ on $[t_0,t]$. Since $z(t_0)=\psi(x(t_0))$ and $\psi$ is injective on $\mathcal{X}$ (Assumption~5), there exists a unique $x(t_0)\in\mathcal{X}$ consistent with the recovered $z(t_0)$. Hence $x(t_0)$ is uniquely determined by
$y(\cdot)$ on $[t_0,t]$.
\end{proof}

\noindent\textbf{Remark 7:}
Koopman-based controllability and observability results assume that a finite-dimensional lifted model exists on $\mathcal{X}$. In practice, $(A_K, B_{K,i}, C_K)$ are learned from data, and their accuracy depends on the choice of $\psi(\cdot)$ and the coverage of the training data.

%\rkb{Still need a use case}

%\rkb{What does this exactly mean for the scenario of a system being approximate as Neural ODE in a practical sense: if we think about observability, we want to reconstruct latent space. Then we want to check whether knowing the output can help in determining the initial hidden state, and whether a neural network can learn a mapping so that a particular $u$ can drive the system of some hidden state to some desired target.} 

%%%%%%%%%%%%%%%%%%%%%%%%%%%%%%%%%%%%%%%%%%%%%%%%%%%%%%%%%%%%%%%%%%%%%%%%%%%%

%%%%%%%%%%%%%%%%%%%%%%%%%%%%%%%%%%%%%%%%%%%%%%%%%%%
\begin{figure}[!b]
	\centering
\includegraphics[width=0.4\textwidth]{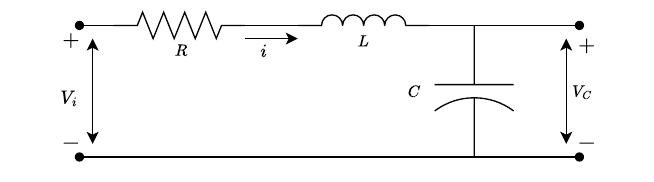}
	\caption{A series RLC circuit where $V_i = 1~\text{V}$ is the input voltage, $V_C$ is the voltage across the capacitor $C = 0.5~\text{F}$, and $i$ is the current flow through resistor $R = 1~\Omega$, inductor $L = 1~\text{H}$, and capacitor $C$.}
	\label{Fig2}
\end{figure}
%%%%%%%%%%%%%%%%%%%%%%%%%%%%%%%%%%%%%%%%%%%%%%%%%%

%---------------------------------------------------------------------------

%%%%%%%%%%%%%%%%%%%%%%%%%%%%%%%%%%%%%%%%%%%%%%%%%%%%%%%%%%%%%%%%%%%%%%%%%%%%
\section{Validation on RLC Circuit}
%%%%%%%%%%%%%%%%%%%%%%%%%%%%%%%%%%%%%%%%%%%%%%%%%%%%%%%%%%%%%%%%%%%%%%%%%%%%

To illustrate the proposed Neural ODE framework, we consider a series RLC circuit as a validation example because its controllability and observability properties are analytically known, as shown in Figure~\ref{Fig2}. This example is used to learn the system dynamics and evaluate empirical controllability and observability. The neural ODE is trained on $1500$ trajectories with $300$ validation samples, where initial conditions are uniformly sampled from $[-1,1]^2$. The time horizon is $t \in [0,10]$ with $400$ discretization points, and the reference trajectories are generated using a fourth-order Runge--Kutta (RK4) solver applied to the series RLC system. The neural ODE uses a fully connected network with three hidden layers of $64$ neurons and \texttt{Tanh} activation functions. The model is trained by minimizing the mean squared error (MSE) between the predicted and true trajectories using the Adam optimizer with a learning rate of $10^{-3}$ for $300$ epochs.

%%%%%%%%%%%%%%%%%%%%%%%%%%%%%%%%%%%%%%%%%%%%%%%%%%%%%%%%%%
\begin{figure}[!t]
    \centering
    \subfloat[]{\includegraphics[width=0.24\textwidth]{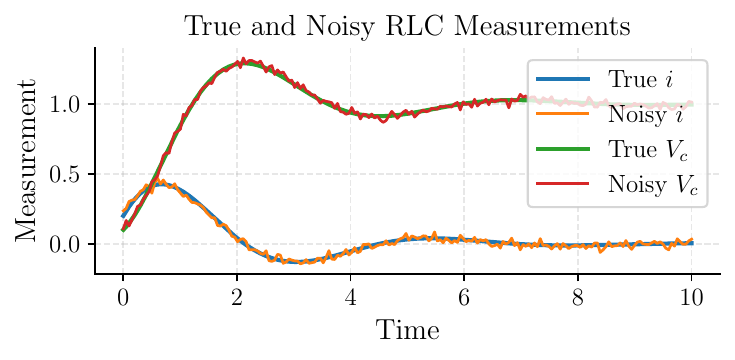}\label{fig:3a}}
    \hfill
    \subfloat[]{\includegraphics[width=0.24\textwidth]{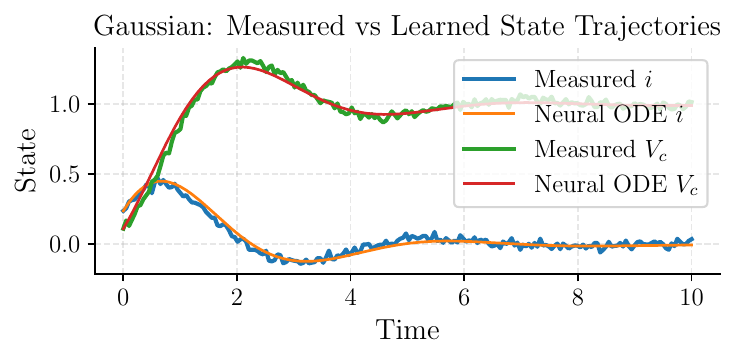}\label{fig:3b}}
    \hfill
    \subfloat[]{\includegraphics[width=0.45\textwidth]{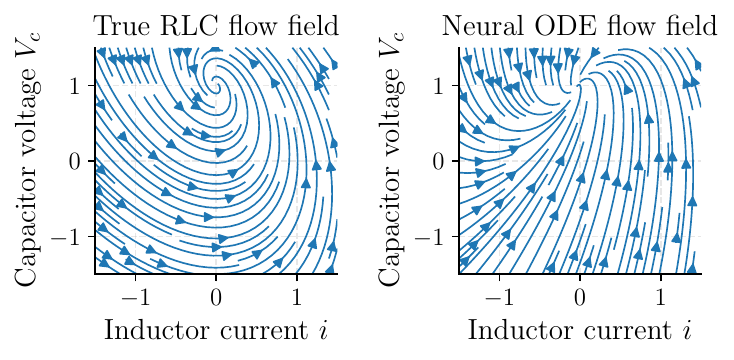}\label{fig:3c}}

\caption{(a) True state trajectories and their Gaussian-noisy counterparts for the series RLC circuit used in training, with noise distributed as $\mathcal{N}(0,\,0.02^2)$; (b) Trajectory fitting of the learned neural ODE for Gaussian-noisy RLC circuit data; (c) Phase-space flow fields of the true RLC system (left) and the learned Neural ODE (right) in the $(i, V_c)$ plane. Flow matching is not explicitly enforced here since the neural ODE is trained to match trajectories \cite{holderrieth2025introduction}.}  
    \label{Fig3}
\end{figure}
%%%%%%%%%%%%%%%%%%%%%%%%%%%%%%%%%%%%%%%%%%%%%%%%%%%%%%%%%%%%%%%
%%%%%%%%%%%%%%%%%%%%%%%%%%%%%%%%%%%%%%%%%%%%%%%%%%%%%%
\begin{figure}[!t]
    \centering
    \subfloat[$W_c^K$ comparison]{%
        \includegraphics[width=0.24\textwidth]
        {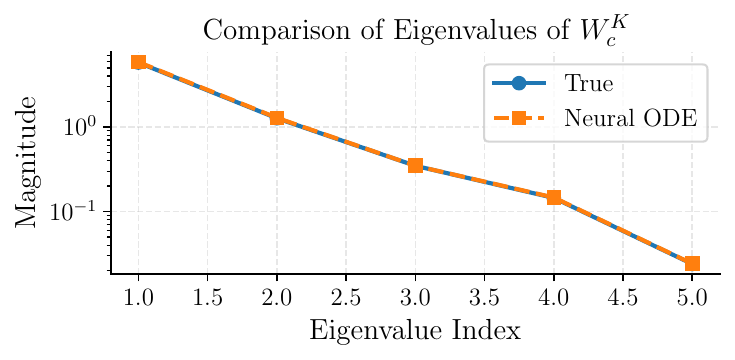}%
        \label{Fig:WcK}}
    \hfill
    \subfloat[$W_o^K$ comparison]{%
        \includegraphics[width=0.24\textwidth]
        {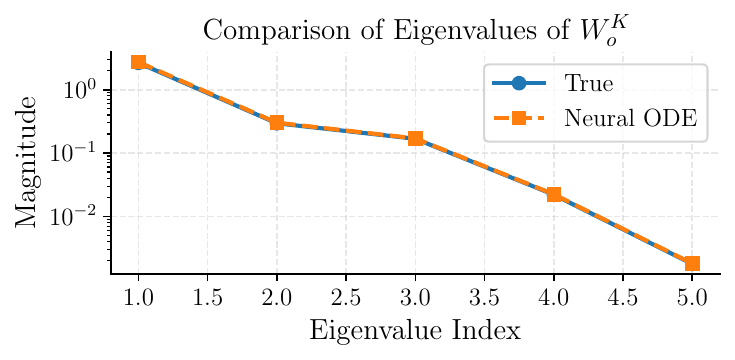}%
        \label{Fig:WoK}}
    \caption{Comparison of the eigenvalues of the Koopman (a) controllability and (b) observability Gramians between the true RLC system and the learned neural ODE. The close agreement suggests that the learned model produces similar lifted controllability and observability characteristics.}
    \label{Fig:KoopmanGramians}
\end{figure}
%%%%%%%%%%%%%%%%%%%%%%%%%%%%%%%%%%%%%%%%%%%%%%%%%%
\begin{figure}[!t]
    \centering
    \subfloat[True system]{%
        \includegraphics[width=0.24\textwidth]
        {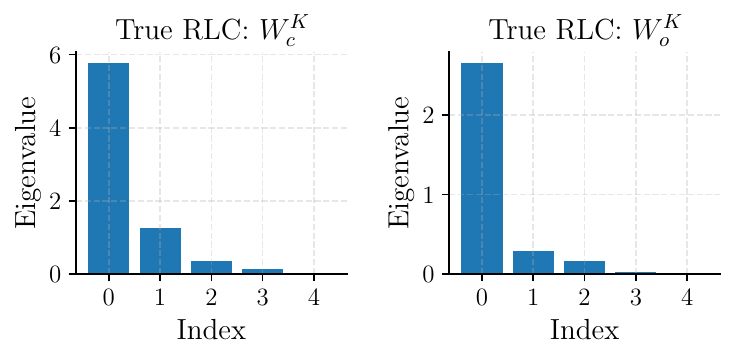}%
        \label{Fig:eig_true}}
    \hfill
    \subfloat[Neural ODE]{%
        \includegraphics[width=0.24\textwidth]
        {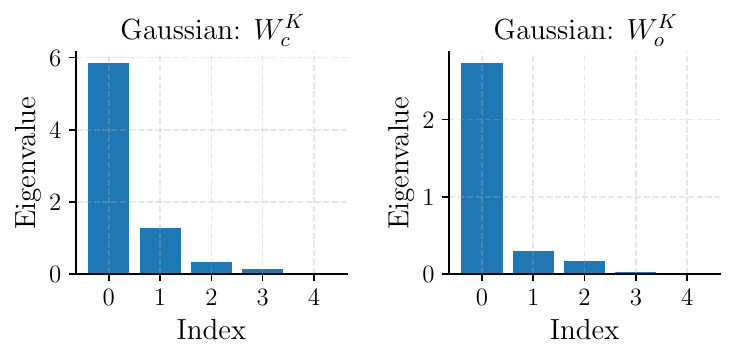}%
        \label{Fig:eig_node}}
    \caption{Eigenvalue spectra of the Koopman controllability $(W_c^K)$ and observability $(W_o^K)$ Gramians for (a) the true RLC system and (b) the learned neural ODE. The eigenvalues are positive and exhibit similar decay patterns, indicating that the learned model captures key controllability and observability characteristics in the lifted space.}
    \label{Fig:KoopmanEigs}
\end{figure}
%==============================================================

The series RLC circuit is described by two first-order differential equations using state variables $x(t) = \begin{bmatrix}  V_c(t) & i (t) \end{bmatrix}^{\top}$ and output variable $y(t)=V_c(t)$ with initial conditions $V_c(0)=V_{c0}$ and $i(0)=i_0$ as follows \cite{alexander2021circuits}:
\begin{equation}
\begin{aligned}\label{eq:23}
\frac{dV_C}{dt} &= \frac{1}{C} i; \qquad
\frac{di}{dt} &= -\frac{1}{L} V_C - \frac{R}{L} i+ \frac{1}{L} V_i(t).
\end{aligned}
\end{equation}
Trajectory data is generated using constant input excitation ($V_i (t) = 1 \text{V}$). Under this condition, the system has a unique equilibrium, so the multi-equilibrium is not discussed here.
Figure~\ref{fig:3a} shows the simulated true trajectories and their Gaussian-noisy measurements used for training. We incorporate Gaussian noise $\mathcal{N}(0,\,0.02^2)$.
The neural ODE is trained with
$\mathcal{L}=\frac{1}{T}\sum_{k=1}^{T} \left\|x_{\text{true}}(t_k)-x_\theta(t_k)\right\|^2$,
where $x_{\text{true}}$ is the true state obtained from simulation of Equation~\eqref{eq:23} with Gaussian noise. The learned neural ODE provides a continuous-time approximation of the RLC dynamics and enables system-theoretic analysis. As shown in Figure~\ref{fig:3b}, the learned neural ODE solution closely follows the true noisy trajectories. The neural ODE produces smoother trajectories by suggesting that the model captures the underlying system dynamics instead of fitting the measurement noise. Based on our example, we find that the neural ODE works as a filter that attenuates the Gaussian noise and provides a solution that is close to the true solution without noise. The vector fields in Figure~\ref{fig:3c} show similar trajectory patterns and directional behavior, resulting in the neural ODE approximating the underlying system dynamics in a qualitative sense. While small differences are present, the overall structure of the flow is consistent.

Koopman-based controllability and observability Gramians are computed for both the true RLC system and the learned neural ODE in the lifted space to validate the proposed framework. A polynomial lifting map $\psi: \mathbb{R}^2\rightarrow \mathbb{R}^5$ is defined as $\psi(x) = \begin{bmatrix}
i & V_c & i^2 & i V_c & V_c^2
\end{bmatrix}^{\top}$ is used to represent the dynamics in a lifted observable space. The resulting Koopman controllability and observability Gramians ($W_c^K$ and $W_o^K$) are computed for both the true system and the learned neural ODE according to Equation~\eqref{eq:WcK} and Equation~\eqref{eq:WoK}, respectively, with finite-difference perturbations with $\epsilon = 0.05$. In this work, we do not explicitly identify Koopman system matrices $(A_K, B_K, C_K)$. Figures~\ref{Fig:WcK} and~\ref{Fig:WoK} show that $W_c^K$ and $W_o^K$ from the neural ODE are consistent with those of the true RLC system. This indicates the learned model preserves the lifted controllability and observability. The eigenvalue spectra in Figures~\ref{Fig:eig_true} and~\ref{Fig:eig_node} are positive in both cases. This is consistent with the full-rank condition in Equation~\eqref{eq:WcK} and Theorem~\ref{thm:koop_obs}; and supports local controllability and observability under the lifting assumption. The eigenvalues of $W_c^K$ for the true system are: $[5.782, 1.269, 0.348, 0.145, 0.024]$; whereas for the neural ODE are: $[5.815, 1.275, 0.347, 0.145, 0.024]$, showing that the neural ODE is consistent with the local controllability of the true RLC system. Similarly, the eigenvalues of $W_o^K$ for the true system are: $[2.659, 0.295, 0.168, 0.022, 0.002]$; whereas for the neural ODE are: $[2.703, 0.299, 0.169, 0.022, 0.002]$, confirming the fidelity of the neural ODE in capturing the observable structure of the underlying RLC system.

%%%%%%%%%%%%%%%%%%%%%%%%%%%%%%%%%%%%%%%%%%%%%%%%%%%%%%%%%%
\begin{figure}[!t]
    \centering
    \subfloat[]{\includegraphics[width=0.24\textwidth]{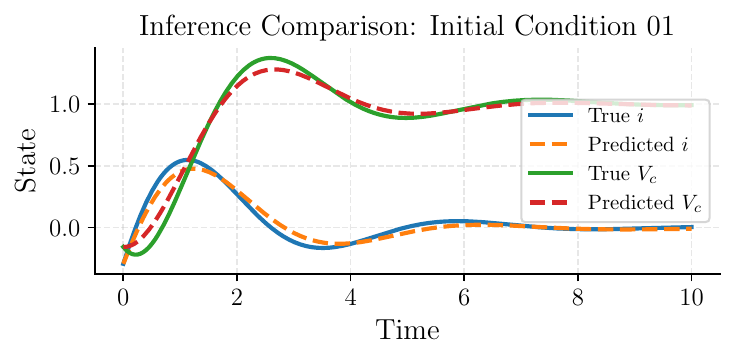}\label{fig6:ic1}}
    \hfill
    \subfloat[]{\includegraphics[width=0.24\textwidth]{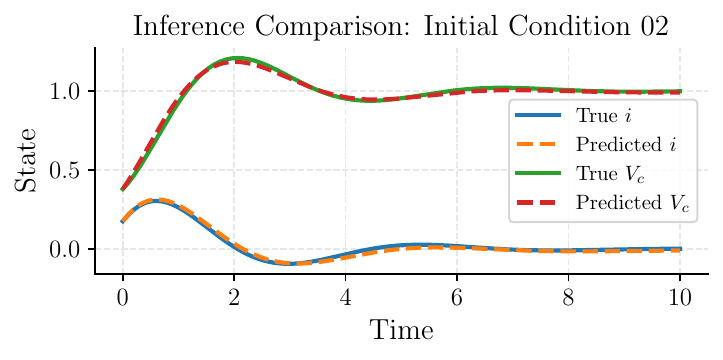}\label{fig6:ic2}}
    \hfill
    \subfloat[]{\includegraphics[width=0.24\textwidth]{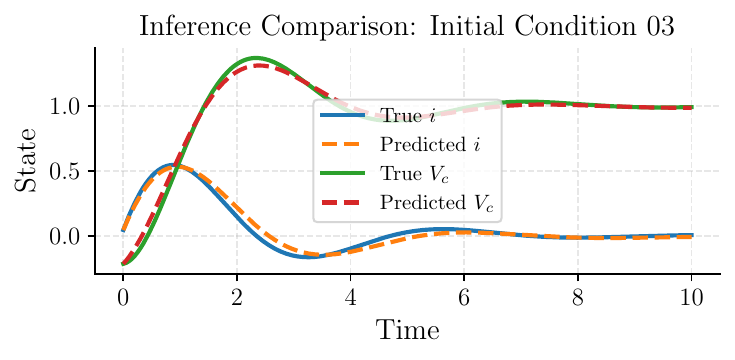}\label{fig6:ic3}}
    \hfill
    \subfloat[]{\includegraphics[width=0.24\textwidth]{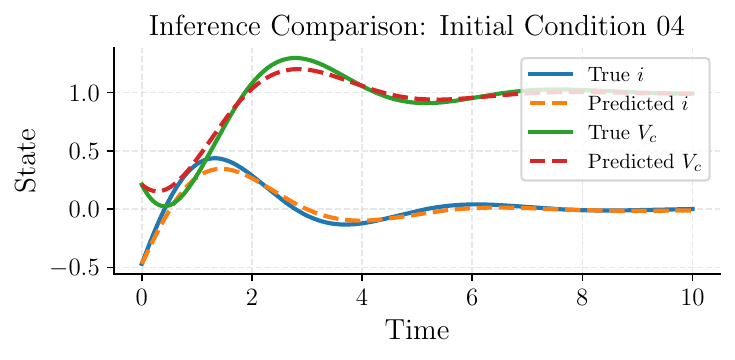}\label{fig6:ic4}}
    \caption{Inference comparison of the Gaussian model for the states $x(t)=\begin{bmatrix} V_c(t) & i(t) \end{bmatrix}^{\top}$ under different initial conditions:(a) $[-0.20,\; 0.50]^\top$, (b) $[0.40,\; -0.30]^\top$,(c) $[0.20,\; 0.80]^\top$, and (d) $[-0.10,\; -0.60]^\top$.}
    \label{Fig6}
\end{figure}
%%%%%%%%%%%%%%%%%%%%%%%%%%%%%%%%%%%%%%%%%
%%%%%%%%%%%%%%%%%%%%%%%%%%%%%%%%%%%%%%%%%
\begin{figure}[!t]
    \centering
    \subfloat[]{\includegraphics[width=0.24\textwidth]{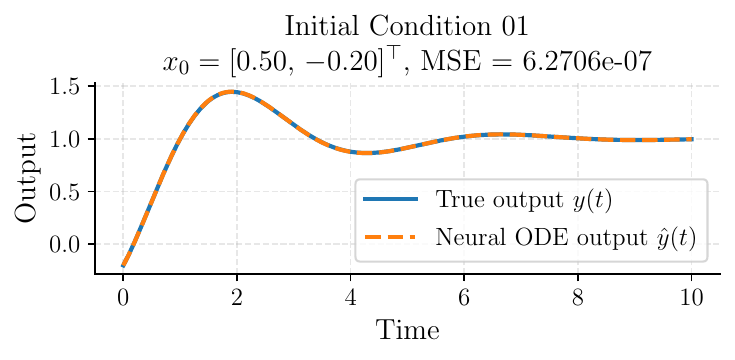}\label{fig:output_ic1}}
    \hfill
    \subfloat[]{\includegraphics[width=0.24\textwidth]{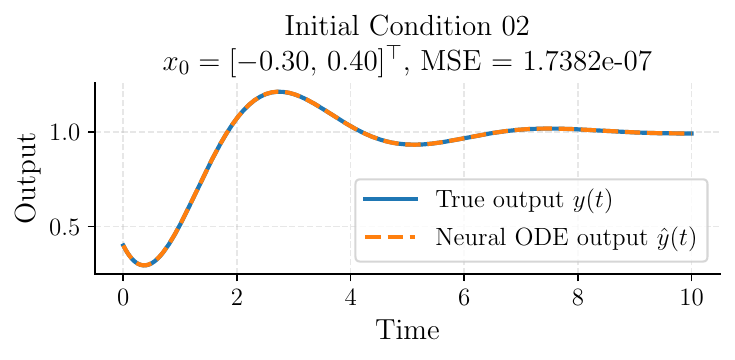}\label{fig:output_ic2}}
    \hfill
    \subfloat[]{\includegraphics[width=0.24\textwidth]{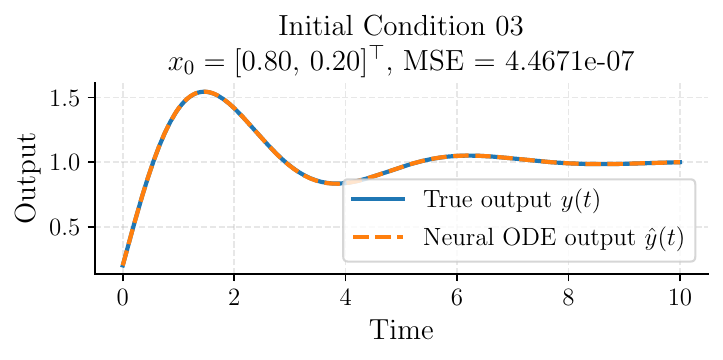}\label{fig:output_ic3}}
    \hfill
    \subfloat[]{\includegraphics[width=0.24\textwidth]{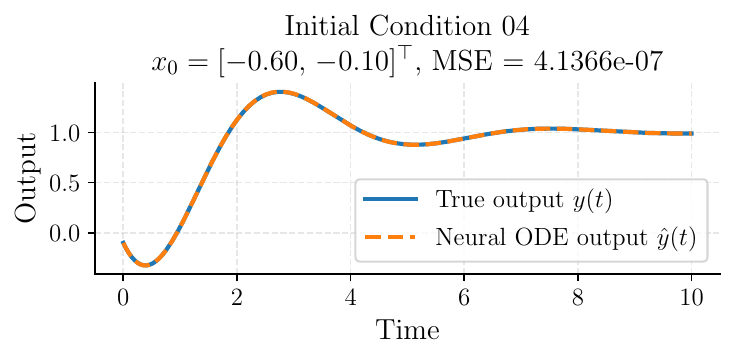}\label{fig:output_ic4}}

    \caption{Comparison of true and neural ODE predicted outputs for four unseen initial conditions: (a) $[0.50,\,-0.20]^\top$, (b) $[-0.30,\,0.40]^\top$, (c) $[0.80,\,0.20]^\top$, and (d) $[-0.60,\,-0.10]^\top$.}
    \label{Fig8}
\end{figure}
%%%%%%%%%%%%%%%%%%%%%%%%%%%%%%%%%%%%%%%%%
The generalization ability of the trained neural ODE model is evaluated by testing on a set of unseen initial conditions. As illustrated in Figures~\ref{fig6:ic1}--\ref{fig6:ic4}, the predicted trajectories nearly follow the true state trajectories across all four tested initial conditions: $[-0.20,\; 0.50]^\top$, $[0.40,\; -0.30]^\top$, $[0.20,\; 0.80]^\top$, and $[-0.10,\; -0.60]^\top$. Hence, the neural ODE generalizes to new operating conditions instead of just fitting the training data. The consistency of tracking performance across diverse initial conditions further indicates the robustness of the neural ODE model to variations in the state space.

At this point, to evaluate output reconstruction, the predicted output, $\hat{y}(t)$ is compared with the true output for multiple initial conditions. As shown in Figures~\ref{fig:output_ic1}--\ref{fig:output_ic4}, the neural ODE output closely follows the true trajectory in all conditions. The MSE values remain small, on the order of $10^{-7}$, indicating nearly accurate reconstruction of the output signal though the trained with Gaussian-noisy data. For example, if we have a close look on Figure \ref{fig:output_ic1}, the first initial condition provides an MSE of $6.2706\times10^{-7}$, while other cases show similar or lower errors. Thus, the predicted outputs nearly match the true trajectories across all initial conditions resulting an accurate output reconstruction and consistent generalization.

The neural ODE understands the system’s behavior, so it can predict new situations correctly and not just repeat what it has seen. It preserves the controllability and observability structures, as shown by the close agreement of the Koopman Gramians with those of the true system. The output trajectories are distinct for different initial conditions and are reconstructed with low mean error, indicating the output retains sufficient information about the system states. The neural ODE also performs well on unseen initial conditions, showing good generalization. Overall, the neural ODE maintains both the system dynamics and input–output behavior. Code, model, and data are available at \url{https://github.com/AARC-lab/CDC2026_NODE}.

%%%%%%%%%%%%%%%%%%%%%%%%%%%%%%%%%%%%%%%%%%%%%%%%%%%%%%%%%%

\section{Conclusion}

In this paper, we work on neural ODEs from a control-theoretic perspective incorporating controllability and observability properties. The proposed framework is evaluated on a series RLC circuit using trajectory data. The proposed approach is system-independent and generalizes to neural ODEs trained on continuous-time dynamical systems. The learned neural ODE reproduces system trajectories under Gaussian measurement noise and shows consistent behavior across unseen initial conditions. Empirical controllability analysis based on trajectory perturbations indicates that the system states can be influenced along multiple directions. Output comparisons across different initial conditions produce distinguishable trajectories, which is consistent with an observability-based interpretation. Therefore, the neural ODE models can capture key dynamical features relevant for control-oriented analysis within a local operating region. These findings are based on trajectory-level analysis and should be interpreted as local properties of the learned dynamics. Moving beyond neural ODEs, in future work, we will explore the use of stochastic differential equations and flow matching to model the evolution of dynamical systems.

%=============================================================
\section*{Author Contributions}

\textbf{Md Saiful Islam:} Conceptualization; methodology; investigation; data curation; software; validation; formal analysis; visualization; and writing--original draft.
\textbf{Rahul Bhadani:} Conceptualization; project administration; resources; supervision; formal analysis; and writing--review \& editing.

\bibliographystyle{IEEEtran}
\bibliography{references}

\end{document}